\documentclass[11pt]{article}

\usepackage{authblk}
\usepackage{amsthm}
\usepackage[letterpaper, left=1in, right=1in, bottom=1in, top=1in]{geometry}

\usepackage{color}
\usepackage{amsmath}
\usepackage{amssymb}
\usepackage{graphicx}
\graphicspath{{img/}}
\usepackage{mathtools}
\DeclarePairedDelimiter{\abs}{\lvert}{\rvert}
\usepackage{lineno,hyperref}
\modulolinenumbers[5]

\usepackage[T1]{fontenc}
\usepackage{diagbox}
\usepackage{url}
\usepackage{hyperref}
\usepackage{cleveref}
\newtheorem{theorem}{Theorem}
\newtheorem{corollary}{Corollary}
\newtheorem{lemma}{Lemma}
\newtheorem{proposition}{Proposition}

\theoremstyle{definition}
\newtheorem{definition}{Definition}

\DeclareMathOperator{\cent}{center}
\newcommand{\LogRmQ}{\textsf{LogRmQ}}
\newcommand\ceil[1]{\lceil#1\rceil}
\newcommand\floor[1]{\lfloor#1\rfloor}
\newcommand{\polylog}{\mathsf{polylog}}
\newcommand{\swpushback}{\mathsf{pushback}}
\newcommand{\swpop}{\mathsf{pop}}
\newcommand{\sups}{\mathsf{sups}}

\newcommand{\SUPS}{\mathsf{SUPS}}
\newcommand{\MUPS}{\mathsf{MUPS}}
\newcommand{\lMUPS}{\mathsf{LMUPS}}
\newcommand{\MLen}{\mathsf{MLen}}
\newcommand{\inc}{\mathsf{inc}}

\begin{document}
\title{Data structures for computing unique palindromes\\ in static and non-static strings}

\author[1]{Takuya~Mieno}
\author[2,3]{Mitsuru~Funakoshi}

\affil[1]{Department of Computer and Network Engineering, University of Electro-Communications, Japan}
\affil[2]{Department of Informatics, Kyushu University, Japan}
\affil[3]{Japan Society for the Promotion of Science}

\date{}
\maketitle

\abstract{
  A palindromic substring $T[i.. j]$ of a string $T$ is said to be a shortest unique palindromic substring~(SUPS) in $T$ for an interval $[p, q]$
  if $T[i.. j]$ is a shortest {palindromic substring}
  such that
  $T[i.. j]$ occurs only once in $T$, and
  $[i, j]$ contains $[p, q]$.
  The SUPS problem is, given a string $T$ of length $n$, to construct a data structure that can compute all the SUPSs for any given query interval.
  It is known that any SUPS query can be answered in $O(\alpha)$ time after $O(n)$-time preprocessing,
  where $\alpha$ is the number of SUPSs to output [Inoue et al., 2018].
In this paper, we first
  show that $\alpha$ is at most $4$, and the upper bound is tight.
We also show that the total sum of lengths of minimal unique 
  {palindromic} 
  substrings of string $T$,
  which is strongly related to SUPSs, is $O(n)$.
  Then, we present the first $O(n)$-bits data structures that can answer any SUPS query in constant time.
  Also, we present an algorithm to solve the SUPS problem for a sliding window
  that can answer any query in $O(\log\log W)$ time and update data structures in amortized 
  {$O(\log\sigma + \log\log W)$ time,}
  where $W$ is the size of the window, and $\sigma$ is the alphabet size.
  Furthermore, we consider the SUPS problem in the after-edit model and present an efficient algorithm.
  Namely, we present an algorithm that uses $O(n)$ time for preprocessing
  and answers any $k$ SUPS queries in $O(\log n\log\log n + k\log\log n)$ time
  after single character substitution.
  Finally, as a by-product,
  we propose a fully-dynamic data structure for range minimum queries (RmQs) with a constraint where
  the width of each query range is limited to poly-logarithmic.
  The constrained RmQ data structure can answer such a query in constant time and support a single-element edit operation in amortized constant time.
}


\section{Introduction}

A substring $T[i.. j]$ of a string $T$ is said to be
a \emph{shortest unique palindromic substring} (in short, \emph{SUPS}) for an interval $[p, q]$
if
$T[i.. j]$ is the shortest substring
such that
$T[i.. j]$ is a palindrome,
$T[i.. j]$ occurs only once in $T$, and
the occurrence contains $[p, q]$, i.e., $[p, q] \subseteq [i, j]$.
The notion of SUPS was introduced by Inoue et al.~\cite{inoue2018algorithms} in 2018, motivated by bioinformatics:
for example, in DNA/RNA sequences, the presence of unique palindromic sequences can affect the immunostimulatory activities of oligonucleotides~\cite{Kuramoto1992oligonucleotide,yamamoto1992unique}.
Given a string $T$ of length $n$, the SUPS problem is to construct a data structure that can compute all SUPSs for any given query interval.
We call this general problem the interval SUPS problem because queries are intervals.
When a query interval is restricted to a single position~(i.e., $p = q$),
the SUPS problem is called the point SUPS problem.
The (interval) SUPS problem was formalized by Inoue et al.~\cite{inoue2018algorithms},
and they showed that all SUPSs for a query interval can be enumerated in $O(\alpha)$ time after $O(n)$-time preprocessing,
where $\alpha$ is the number of SUPSs to output.
Watanabe et al.~\cite{WatanabeNIBT20} considered the SUPS problem on run-length encoded strings to reduce the space usage.
They proposed an $O(r)$-space data structure that can enumerate all SUPSs for a query interval in $O(\sqrt{\log r / \log \log r}+\alpha)$ time
where $r$ is the size of the run-length encoded string, which satisfies $r \le n$.

Both of the above results are for a static string.
It is a natural question whether we can compute SUPSs efficiently in a \emph{dynamic} string.
  In fact, since DNA sequences contain errors and change dynamically, it is worthwhile to consider them in a dynamic string setting.
However, there is no research for solving the SUPS problem on a dynamic string to the best of our knowledge.
Thus, in this paper, as a first step to designing dynamic algorithms, we consider the problem on two \emph{semi-dynamic} models:
the \emph{sliding-window} model and the \emph{after-edit} model.
The sliding-window model aims to compute some objects (e.g., data structure, compressed string, statistics, and so on)
w.r.t. the window sliding over the input string left to right.
The after-edit model aims to compute some objects w.r.t. the string after applying an edit operation to the input string.
Edit operations are given as queries, and they are discarded 
{after processing the query.}
As related work,
the set of \emph{minimal unique palindromic substrings} (\emph{MUPSs}) can be maintained efficiently in the sliding-window model~\cite{mieno2021eertree}.
Also, the set of MUPSs can be updated efficiently in the after-edit model~\cite{funakoshiMUPSafteredit}.
Since MUPSs are strongly related to SUPSs, we utilize the above known results for MUPSs as black boxes.

Contributions of this paper are summarized as follows:
\begin{description}
  \item[{Section~\ref{sec:comb}:}] {\bf Combinatorial properties on SUPSs and MUPSs.}
    \begin{itemize}
      \item We show that the number $\alpha$ of SUPSs for any single interval is at most four,
        and the upper bound is tight even for binary strings,
      \item We show that the sum of lengths of MUPSs of a string of length $n$ is $O(n)$.
    \end{itemize}
  \item[{Section~\ref{sec:compact}}:] {\bf Compact SUPS data structures for static strings.}
    \begin{itemize}
      \item We propose a compact data structure of size $3n+2m+o(n)$ bits that can answer any \emph{interval SUPS} query in constant time
        where $n$ is the length of the input string and $m$ is the number of MUPSs of the input string.
      \item We propose a compact data structure of size $3n+m+o(n)$ bits that can answer any \emph{point SUPS} query in constant time.
    \end{itemize}
  \item[{Section~\ref{sec:dynamic_algo}}:] {\bf Algorithms for SUPS problem for semi-dynamic strings.}
    \begin{itemize}
      \item We propose a data structure of size $O(W)$ for the sliding-window SUPS problem
        that supports SUPS query in $O(\log\log W)$ time and each window-shift in amortized
        $O(\log\sigma + \log\log W)$ time
        where $W$ is the size of the window and $\sigma$ is the alphabet size.
      \item We propose a data structure of size $O(n)$ for the 
{after-substitution SUPS problem}
        that can answer any {after-substitution-SUPS} query in amortized $O(\log n\log\log n)$ time
        after some character in the input string is substituted with another character.
    \end{itemize}
\end{description}
Furthermore, as a by-product, we propose a fully-dynamic data structure
for the range minimum query (RmQ) in which the width of each query range is in $O(\polylog(n))$.
The data structure can answer such a query in constant time and update in (amortized) constant time for any single-element edit operation.
Note that, for the original RmQ without any additional constraint,
it is known that we need $\Omega(\log n/ \log\log n)$ time for answering a query when $O(\polylog(n))$ updating time is allowed~\cite{alstrup1998}.
\paragraph{\bf Related Work.}
A typical application to the sliding-window model is string compression such as LZ77~\cite{LZ77} and PPM~\cite{PPM}.
The sliding-window LZ77 compression is based on the sliding-window suffix tree~\cite{FialaGreene1989,larsson1996extended,ukkonen1995line,senft2005suffix}.
Also, the sliding-window suffix tree can be applied to compute minimal absent words~\cite{crochemore2020absent} and minimal unique substrings~\cite{mieno2021computing},
which are significant concepts for bioinformatics,
in the sliding-window model.
Recently, the sliding-window palindromic tree was proposed, and it can be applied to compute MUPSs in the sliding-window model~\cite{mieno2021eertree}.

The after-edit model was formalized by Amir et al.~\cite{Amiretal17} in 2017.
They tackled the problem of computing
{the longest common substring}
for two strings in the after-edit model,
and proposed an algorithm running in poly-logarithmic time.
Afterward, Abedin et al.~\cite{Abedinetal22} improved the complexities.
Also, the problems of computing the longest Lyndon substring~\cite{Urabe18}, the longest palindrome~\cite{Funakoshi21}, and the set of MUPSs~\cite{funakoshiMUPSafteredit} were considered in the after-edit model.

As for more general settings,
Amir et al.~\cite{Amiretal20} proposed a fully-dynamic algorithm for computing 
{the longest common substrings} 
for two dynamic strings.
They also developed a general (probabilistic) scheme for dynamic problems on strings
and applied it to the computation of the longest Lyndon substring and the longest palindrome in a dynamic string.
Besides that, there are several studies for dynamic settings (e.g.,~\cite{GawrychowskiKKL18,AmirBCK19,Charalampopoulos20}).
In particular, a fully-dynamic and deterministic algorithm for computing the longest palindrome was shown in~\cite{Amir_Boneh_19}.

\paragraph{\bf Paper Organization.}
The rest of this paper is organized as follows:
In Section~\ref{sec:pre}, we give basic notations and algorithmic tools.
In Section~\ref{sec:static_algo}, we review a known static SUPS data structure proposed by Inoue at al.~\cite{inoue2018algorithms}, which is the basis of most of our methods.
In Section~\ref{sec:comb}, 
{we investigate combinatorial properties on SUPSs and MUPSs.}
We show the tight bounds on the maximum number of SUPSs for an interval and an upper bound of the total sum of the lengths of MUPSs.
Further, we propose a simple algorithm for point SUPS queries based on the combinatorial results.
In Section~\ref{sec:compact},
we propose the first compact data structures that can answer any SUPS query in output-sensitive time.
In Section~\ref{sec:dynamic_algo}, we consider how to update SUPS data structures in semi-dynamic settings and propose efficient algorithms.
Finally, in Section~\ref{sec:conclusions}, we conclude our paper and discuss future work.
 \section{Preliminaries}\label{sec:pre}
\subsection{Strings}
Let $\Sigma$ be an alphabet.
An element of $\Sigma$ is called a character.
An element of $\Sigma^\ast$ is called a string.
The length of a string $T$ is denoted by $\abs{T}$.
The empty string $\varepsilon$ is the string of length $0$.
For each $i$ with $1\le i \le \abs{T}$, we denote by $T[i]$ the $i$-th character of $T$.
If $T = xyz$, then $x$, $y$, and $z$ are called a prefix, substring, and suffix of $T$, respectively.
For each $i, j$ with $1\le i \le j \le \abs{T}$, we denote by $T[i..j]$ the substring of $T$
starting at position $i$ and ending at position $j$.
For convenience, let $T[i'.. j'] = \varepsilon$ for any $i', j'$ with $i' > j'$.
We say that string $w$ is unique in $T$ if $w$ occurs only once in $T$.
For convenience, we define that the empty string $\varepsilon$ is not unique in any string.
For a {non-empty} string $T$ and a positive integer $p$ with $p \le \abs{T}$,
the integer $p$ is a period of $T$ if $T[i] = T[i+p]$ holds for every $i$ with $1\le i \le \abs{T}-p$.
We also say that $T$ has a period $p$ if $p$ is a period of $T$.
{For convenience, the empty string $\varepsilon$ is defined to have period $0$.}

Let $T^R$ denote the reversal of a string $T$, i.e., $T[i] = T^R[n-i+1]$ for every $i$ with $1\le i \le n$.
A string $P$ is called a palindrome if $P = P^R$ holds.
A palindrome $P$ is called an even-palindrome (resp.,~odd-palindrome) if $\abs{P}$ is even (resp.,~odd).
The length-$\ceil{\abs{P}/2}$
prefix (resp.,~suffix) of a palindrome $P$ is called the left arm (resp.,~right arm) of $P$.
Let $w = T[i.. j]$ be a palindromic substring of $T$.
The center of $w$ is $(i+j)/2$ and is denoted by $\cent(w)$.
For a non-negative integer $\ell$, $x = T[i-\ell.. j+\ell]$ is said to be an expansion of $w$ if
$1 \le i-\ell \le j+\ell \le n$ and $x$ is a palindrome.
Also, $T[i+\ell.. j-\ell]$ is said to be a contraction of $w$.
Further, if $i = 1$, $j = n$, or $T[i-1] \ne T[j+1]$, then $w$ is said to be a maximal palindrome.

A palindromic substring $u = T[i.. j]$ of a string $T$ is said to be a minimal unique palindromic substring~(MUPS) in $T$
if $u$ is unique in $T$ and $T[i+1.. j-1]$ is not unique in $T$.
A palindromic substring $v = T[i.. j]$ of a string $T$ is said to be a shortest unique palindromic substring~(SUPS) for an interval $[p, q]$ in $T$
if $v$ is unique in $T$, the occurrence contains interval $[p, q]$, and any shorter palindromic substring of $T$ that contains $[p, q]$ is not unique in $T$.
We denote by $\SUPS_T([p, q])$ the set of SUPSs for $[p, q]$.
Note that all palindromes in $\SUPS_T([p, q])$ have equal lengths.
See also Fig.~\ref{fig:MUPS_SUPS} for examples.
\begin{figure}[t]
  \centerline{
    \includegraphics[width=0.8\linewidth]{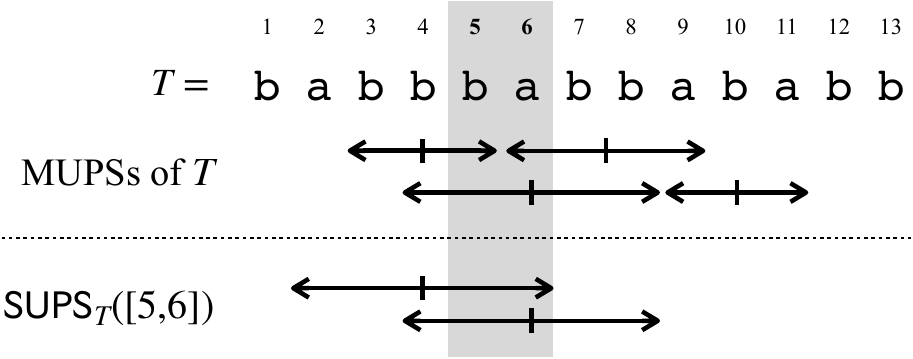}
  }
  \caption{
    MUPSs of string $T = \mathtt{babbbabbababb}$ are
    $\mathtt{bbb}$,
    $\mathtt{bbabb}$,
    $\mathtt{abba}$, and
    $\mathtt{aba}$.
    SUPSs for interval $[5, 6]$ in $T$ are
    $T[2.. 6] = \mathtt{abbba}$ and $T[4.. 8] = \mathtt{bbabb}$.
    The first SUPS $T[2.. 6]$ is an expansion of MUPS $T[3.. 5] = \mathtt{bbb}$,
    and the second SUPS $T[4.. 8]$ itself is a MUPS.
  }
  \label{fig:MUPS_SUPS}
\end{figure}
The interval SUPS problem is,
given a string $T$ for preprocessing and an interval $[p, q]$ as a query,
to compute $\SUPS_T([p, q])$.
We sometimes simply refer to the interval SUPS problem as the SUPS problem.
When every query interval is restricted to a single position~(i.e., $p = q$), the SUPS problem is called the point SUPS problem.

In what follows, we fix a string $T$ of arbitrary length $n > 0$ over an integer alphabet of size $\sigma = O(\mathsf{poly}(n))$.
Also, our computational model is a standard word RAM model of word size $\Omega(\log n)$.

\subsection{Periodicity of Palindromic Suffixes}\label{subsec:period}
In this subsection, 
{we recall some properties regarding the periodicity} 
of palindromic suffixes of $T[1..i]$ that we use.
Let $\mathbf{S}_i = \{ w_1, \ldots, w_{g}\}$ be the set of lengths of palindromic suffixes of $T[1..i]$,
where $g$ is the number of palindromic suffixes of $T[1..i]$
and 
{$w_{k-1} < w_k$} 
for $2 \leq k \leq g$.
Let $d_k$ be the progression difference for $w_k$,
i.e., $d_k = w_{k} - w_{k-1}$ for $2 \leq k \leq g$.
For convenience, 
{let $d_1 = 0$.}

Then, the following results are known:

\begin{lemma}[\cite{ApostolicoBG95,GasieniecSWAT96,matsubara_tcs2009}]
  \label{lem:palindromic_suffixes}
  \hfill
  \begin{enumerate}
    \item[(A)] For any $1 \leq k < g$, $d_{k+1} \geq d_{k}$.
    \item[(B)] For any $1 < k < g$, if $d_{k+1} \neq d_{k}$, then $d_{k+1} \geq d_{k} + d_{k-1}$.
    \item[(C)] $\mathbf{S}_i$ can be represented by $O(\log i)$ arithmetic progressions, where each arithmetic progression is a tuple $\langle s, d, f \rangle$ representing the sequence $s, s+d, \ldots, s + (f-1)d$ of lengths of $f$ palindromic suffixes with common difference $d$.
    \item[(D)] The common difference $d$ is 
    {the smallest period of all} 
    palindromic suffixes of $T[1..i]$ whose length belongs to the arithmetic progression $\langle s, d, f \rangle$.
  \end{enumerate}
\end{lemma}

\begin{lemma}[\cite{matsubara_tcs2009}]\label{lem:batched_extension}
{For any $\langle s, d, f \rangle$ from the representation of palindromic suffixes of $T[1..i]$,}
  there exist palindromes $u, v$ and a non-negative integer $q$,
  such that $(uv)^{f+q-1} u$ (resp. $(uv)^q u$) is
  the longest (resp. shortest) palindromic 
  {suffix} 
  represented by $\langle s, d, f \rangle$
  with $\abs{uv} = d$.
\end{lemma}

For a position $i$, divide the set of palindromic suffixes of $T[1..i]$ into groups $G_1, G_2, \ldots, G_{\pi}$ 
{w.r.t.} 
their smallest periods in increasing order.
For each $G_r = \langle s_r, d_r, f_r \rangle$ with $1 \leq r \leq \pi$,
let $u_r$ and $v_r$ be the corresponding variables
used in Lemma~\ref{lem:batched_extension}.

Let $\mathit{lcp}(x,y)$ for strings $x$ and $y$ denote the length of the longest common prefix of $x$ and $y$.
Also, let $\alpha_r = \mathit{lcp}((T[1..i-s_r])^R, T[i+1..n])$
and $\beta_r = \mathit{lcp}((T[1..i-s_r-(f_r-1)d_r])^R, T[i+1..n])$ {if $f_r \ge 2$}.
Namely, $s_r + 2\alpha_r$ (resp. $s_r+(f_r-1)d_r + 2 \beta_r$) is the length of the maximal expansion of the shortest (resp. longest) palindrome of $\langle s_r, d_r, f_r \rangle$.
See also Figure~\ref{fig:group} for a concrete example of $G_r$.
\begin{figure}[t]
  \centerline{
    \includegraphics[width=0.6\linewidth]{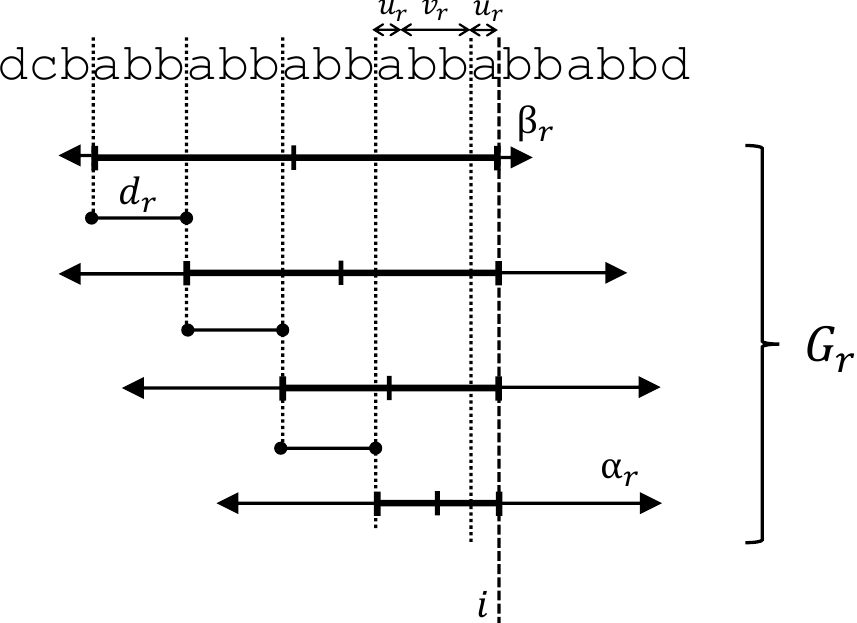}
  }
  \caption{
    Example for a group $G_r$, with string $\mathtt{dcbabbabbabbabbabbabbd}$.
    Here palindromic suffixes $(\mathtt{abb})^j \mathtt{a}$ with $1\leq j \leq 4$ ending at $i$ belong to $G_r = \langle 4, 3, 4 \rangle$.
    Also, $u_r = \mathtt{a}$, $v_r = \mathtt{bb}$, $\alpha_r = 5$, and $\beta_r = 1$ hold.
    Note that this $G_r$ is of {type 2}.
  }
  \label{fig:group}
\end{figure}

{If $f_r = 1$, i.e., if $G_r$ is a singleton,} let $\beta_r$ be the length of the maximal expansion of the palindrome.
{In addition, if $r \ge 2$, let $\alpha_{r}=\beta_{r-1}$.}
{Note that $G_1 = \{\varepsilon\}$ is a singleton, i.e., $f_1 = 1$. For convenience, let $\alpha_1 = 0$.}
Each group $G_r$ is said to be of 
{\emph{type 1} (resp. {\emph{type 2}})}
if $\alpha_r < d_r$ (resp. $\alpha_r \geq d_r$).
Let $k$ be the largest index of groups such that $G_k$ is of type 2.
{Since $\alpha_1 = d_1 = 0$ always holds, $G_1$ is of type 2 and $k$ is well-defined.}
In the proof of Claim (1) and (2) of~\cite{Funakoshi21}, the following statements are also proven:

\begin{corollary}\label{cor:claims}
  The following two statements hold:
  \begin{enumerate}
    \item[(A)] Any expansion of a palindrome $Q$ in group $G_r$ for every $r$ with $1 \le r \le k-1$ except for $u_kv_ku_k$ and $u_k$ cannot be longer than $\abs{Q}+2d_k$.
    \item[(B)]  For every $r$ with $k+1 \le r \le \pi-2$ and every palindrome $P$ in group $G_r$, the length of any expansion of $P$ is at most $\abs{P}+2\beta_r < \abs{P} + 2 d_{r+1}$.
  \end{enumerate}
\end{corollary}

From Corollary~\ref{cor:claims}, the following lemmas can be obtained.
\begin{lemma}\label{lem:type2_expansion}
  Any expansion of a palindrome $Q$ in group $G_r$ for every $r$ with $1 \le r \le k-1$ except for $u_kv_ku_k$ and $u_k$ cannot be unique in $T$.
\end{lemma}
\begin{proof}
  Since $G_k$ is of {type 2}, the length of the maximal expansion of $s_k$ is $\abs{s_k} + 2 \alpha_k \geq \abs{s_k} + 2 d_k$.
  From statement (A) of Corollary~\ref{cor:claims},
  {
    any expansion of a palindrome $Q$, whose center is differ from the center of $s_k$, is contained by the maximal expansion of $s_k$.
    Namely, any expansion of $Q$ occurs at least twice in $s_k$.
  }
\end{proof}
  
\begin{lemma}\label{lem:type1_expansion}
  For every $r$ with $k+1 \le r \le \pi-2$ and every palindrome $P$ in group $G_r$,
  any expansion of $P$ is not longer than 
  $s_{r+2}${, that is,} 
  the length of the shortest palindrome in $G_{r+2}$.
\end{lemma}
\begin{proof}
  From the formula of statement (B) of Corollary~\ref{cor:claims} and the definition of the progression differences $d_{r+1}$ and $d_{r+2}$, $\abs{P}+2\beta_r < \abs{P} + 2 d_{r+1} < \abs{P} + d_{r+1} + d_{r+2} \leq s_{r+1} + d_{r+2} \leq s_{r+2}$ holds.
\end{proof}

\subsection{Tools}
\paragraph*{\bf Longest Common Extension.}
A longest common extension (in short, LCE) query on string $T$ is,
given two integers $i, j$ with
{$1 \le i, j \le n$,}
to compute the length of the longest common prefix (LCP) of two suffixes $T[i.. n]$ and $T[j.. n]$.
It is known~(e.g.,~\cite{Gusfield1997})
that any LCE query can be answered in constant time
using the suffix tree of $T\$$ enhanced with a lowest common ancestor data structure,
where $\$$ is a special character that is not in $\Sigma$.
Once we build an LCE data structure on string $T\#T^R\$$,
we can answer any LCE query in any direction on $T$ in constant time,
where $\# \not\in\Sigma$ is another special character.
Namely, we can compute in constant time the length of
(1) the LCP length of any two suffixes of $T$,
(2) the LCP length of the reverses of any two prefixes of $T$, and
(3) the LCP length of any suffix of $T$ and the reverse of any prefix of $T$.
We call such a data structure a bidirectional LCE data structure.

\paragraph*{\bf RmQ, Predecessor and Successor.}

A range minimum query (RmQ) on integer array $A$ is,
given two indices $i, j$ on $A$ with $i \le j$,
to compute 
{an arbitrary index $k$ such that $A[k]$ is the minimum value from $A[i.. j]$.}

A predecessor (resp., successor) query on non-decreasing integer array $B$ is,
given an integer $x$,
to compute the maximum (resp., minimum) value that is smaller (resp., greater) than $x$.
We use the famous van Emde Boas tree data structure~\cite{vEBtree1977} to answer predecessor/successor queries.
Namely, we can answer a query and update the data structure in $O(\log\log U)$ time on a dynamic array,
where $U$ is the universe size. Also, the space complexity is $O(U)$.
Throughout this paper, 
{we will only apply this result to the case of $U = n$.}

\subsection{Our Problems}
This paper handles SUPS problems under two variants of semi-dynamic models:
the sliding-window model and the after-edit model.
The sliding-window SUPS problem is to support any sequence of queries that consists of the following:
\begin{itemize}
  \item $\swpushback(c)$: append a character $c$ to the right end of the string.
  \item $\swpop()$: remove the first character from the string.
  \item $\sups([p, q])$: output all SUPSs of the string for an interval $[p, q]$.
\end{itemize}
The {after-substitution SUPS} problem on a string $T$ is, given
a substitution operation and a sequence of intervals,
to compute SUPSs of $T'$ for each interval where $T'$ is the string
after applying the substitution \emph{to the original string $T$}.
Note that each substitution is discarded after the corresponding SUPS queries are answered.

From the point of view of how the string changes, there are differences between the above two problems.
On the one hand, in the sliding-window SUPS problem,
the string can be changed dynamically under the constraints of positions to be edited.
On the other hand, in the {after-substitution SUPS problem}, any position of the string can be changed,
however, the string returns to the original one after the SUPS queries are answered.
 
\section{{Inoue et al.'s} Static SUPS Data Structure} \label{sec:static_algo}

The SUPS data structure proposed in \cite{inoue2018algorithms} consists of the following:

\begin{itemize}
  \item the set of MUPSs of $T$,
  \item the set of maximal palindromes of $T$\\
    (or a bidirectional LCE data structure on $T$, instead),
  \item a successor data structure on the starting positions of MUPSs,
  \item a predecessor data structure on the ending positions of MUPSs, and
\item an RmQ data structure on
    {the array $\mathsf{MUPSlen}$ of lengths of MUPSs sorted by their starting positions\footnote{Since MUPSs cannot be nested~\cite{inoue2018algorithms}, they are also sorted by their ending positions.}.}
\end{itemize}

Given a query interval $[p, q]$, we can compute all SUPSs for $[p, q]$ as follows:
First, we determine whether the interval $[p, q]$ covers some MUPS or not
by querying the predecessor of $q$ on the ending positions of MUPSs
and the successor of $p$ on the starting positions of MUPSs.
If $[p, q]$ covers only one MUPS,
the shortest expansion of the MUPS that covers $[p, q]$
is the only SUPS for $[p, q]$ if such a palindrome exists,
and there are no SUPSs for $[p, q]$ otherwise.
If $[p, q]$ covers more than one MUPS,
then there are no SUPSs for $[p, q]$ since any SUPS covers exactly one MUPS~\cite{inoue2018algorithms}.
Otherwise, i.e., if $[p, q]$ covers no 
{MUPSs}, all SUPSs are categorized into following three types (see also Fig.~\ref{fig:sups_cand}):

\begin{figure}[t]
  \centerline{
    \includegraphics[width=0.7\linewidth]{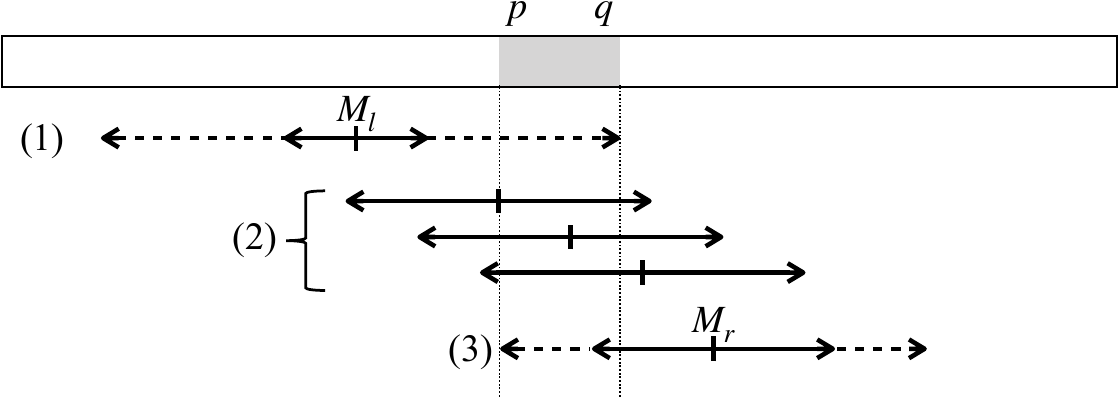}
  }
  \caption{
    Illustration for candidates for SUPSs for interval $[p, q]$.
    Solid arrows represent MUPSs, and dashed arrows represent expansions of MUPSs.
    Note that dashed arrows may not be palindromes in general.
  }
  \label{fig:sups_cand}
\end{figure}

\begin{enumerate}
  \item[(1)] an expansion of the rightmost MUPS $M_l$ which ends before $q$,
\item[(2)] {a MUPS which covers $[p, q]$}, or
  \item[(3)] an expansion of the leftmost MUPS $M_r$ which begins after $p$.
\end{enumerate}

We call $M_l$ and $M_r$ the left-neighbor MUPS and the right-neighbor MUPS of the interval $[p, q]$, respectively.
We can find $M_l$ by querying the predecessor of $q$.
Also, we can determine whether there is an expansion of $M_l$, which covers $[p, q]$
by looking at the maximal palindrome centered at $c_l = \cent(M_l)$.
Precisely, if the maximal palindrome centered at $c_l$ covers $[p, q]$,
its shortest contraction covering $[p, q]$ is the only candidate of type (1).
Otherwise, there is no SUPS of type (1).
We emphasize that we can also compute the maximal palindrome centered at $c_l$
by querying bidirectional LCE once, without the precomputed maximal palindromes.
The candidate of type (3) can be treated similarly.
Finally, all SUPSs of type (2) can be computed by querying RmQ recursively on {the array $\mathsf{MUPSlen}$}.
Let $[b_i, e_i], \ldots, [b_j, e_j] \in \MUPS(T)$ be all MUPSs covering $[p, q]$.
Recall that such range $[i, j]$ of MUPSs can be detected by querying predecessor and successor (see above).
We query the RmQ on $\mathsf{MUPSlen}$ for the range $[i, j]$, and obtain the index $k$ that is the answer of the RmQ.
Namely, $[b_k, e_k]$ is a shortest one within $[b_i, e_i], \ldots, [b_j, e_j]$.
Then, we further query the RmQ on $\mathsf{MUPSlen}$ for the ranges $[i, k-1]$ and $[k+1, j]$, and repeat it recursively
while obtained MUPS is a SUPS for $[p, q]$.
The above operations can be done in time linear in the number of SUPSs to output
by using linear size data structures of predecessor, successor, and RmQ.
 \section{Combinatorial Properties on SUPSs} \label{sec:comb}
\subsection{Tight Bounds on Maximum Number of SUPSs for Single Query} \label{subsec:num}

In this subsection, we prove the following theorem:

\begin{theorem}\label{thm:num_of_SUPS}
  For any interval $[p, q]$ over $T$, the inequality 
  {$\abs{\SUPS_T([p, q])}\le 4$} 
  holds.
  Also, this upper bound is tight even for binary strings.
\end{theorem}

{To begin with,} 
we prove Lemma~\ref{lem:overlap_pal}.
Roughly speaking, Lemma~\ref{lem:overlap_pal} states that
{a periodic structure occurs} 
when two palindromes overlap enough.
Essentially, Lemma~\ref{lem:overlap_pal} has been proven in Lemma 3.3 of~\cite{ApostolicoBG95}.
However, we restate the proposition in a form that is convenient for us and show it for completeness.
\begin{lemma}\label{lem:overlap_pal}
  Let $x = T[i.. i+\ell-1]$ and $y = T[j.. j+\ell-1]$ be palindromic substrings of length $\ell$ of string $T$ with $i < j$.
  If $x$ and $y$ overlap, then $z = T[i.. j+\ell-1]$ has period $2d$ where $d$ is the distance between their center positions.
\end{lemma}
\begin{proof}
Firstly, $d = (2j+\ell-1)/2-(2i+\ell-1)/2=j-i$ holds.
Let $z = rst$ where $r = T[i.. j-1]$, $s = T[j.. i+\ell-1]$, and $t = T[i+\ell.. j]$.
  Since $x$ and $y$ are palindromes, $s^R$ is a prefix of $x$ and a suffix of $y$.
  Namely, $s^R$ is both a prefix and a suffix of $z$, and thus, $z$ has period $\abs{z}-\abs{s^R} = (j-i+\ell)-(i-j+\ell) =2(j-i) = 2d$.
\end{proof}

Now we are ready to prove Theorem~\ref{thm:num_of_SUPS}.

\begin{proof}[Proof of Theorem~\ref{thm:num_of_SUPS}]
  {Let us focus on the SUPSs whose center} 
  is at most $p$.
{We assume we have at least three such SUPSs for a single query interval $[p, q]$ and show that this leads to a contradiction.}
  {Let $\ell$ be the length of the SUPSs.}
  Let $x$, $y$, and $z$ be the SUPSs from left to right,
  and let $c_x$, $c_y$, and $c_z$ be their center positions~(see also Fig.~\ref{fig:three_pals}).
\begin{figure}[t]
  \centerline{
    \includegraphics[width=0.6\linewidth]{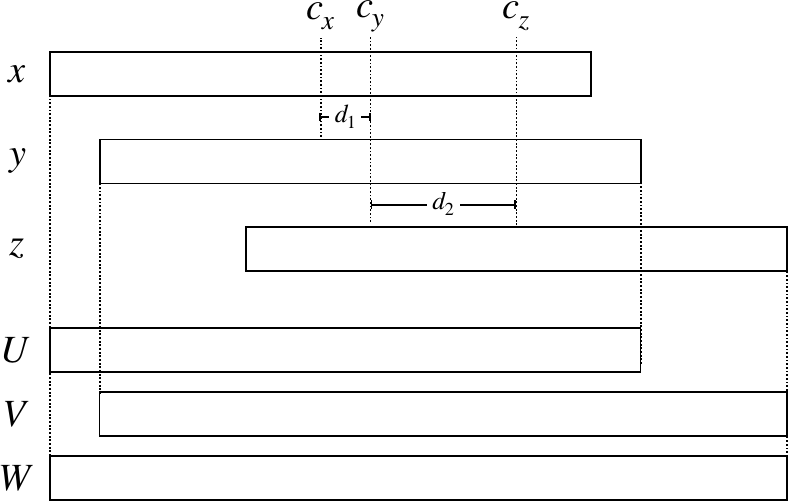}
  }
  \caption{
    Illustration for three overlapped palindromes $x$, $y$, and $z$.
  }
  \label{fig:three_pals}
\end{figure}
  Further let $d_1 = c_y-c_x$ and $d_2 = c_z-c_y$.
Since the center positions of the three SUPSs are at most $p$, and they cover the position $p$,
  they overlap at least $\ell/2$ each other.
  Namely, $d_1 \le \ell/2$, $d_2 \le \ell/2$, and $d_1 + d_2 \le \ell/2$ hold. 
  Next, let
  $U = T[\ceil{c_x-\ell/2}.. \floor{c_y+\ell/2}]$,
  $V = T[\ceil{c_y-\ell/2}.. \floor{c_z+\ell/2}]$, and
  $W = T[\ceil{c_x-\ell/2}.. \floor{c_z+\ell/2}]$.
  By Lemma~\ref{lem:overlap_pal}, $U$ has period $2d_1$ and $V$ has period $2d_2$.
  Thus, $y$ has periods both $2d_1$ and $2d_2$.
  Also, since $2d_1+2d_2 \le \ell$ holds, $y$ has a period $g = \gcd(2d_1, 2d_2)$ by the periodicity lemma~\cite{fine1965uniqueness}
  where $\gcd(a, b)$ denotes the greatest common divisor of $a$ and $b$.
  Then $W$ also has period $g$
  since $g < \abs{y} = \ell$ and $g$ divides both period $2d_1$ of $U$ and period $2d_2$ of $V$.
Furthermore, since $g \leq \min(2d_1, 2d_2) \leq d_1 + d_2$
  and $d_1+d_2+\ell = \abs{W}$, the inequality $g + \ell \leq \abs{W}$ holds,
  and thus, $x = W[1.. \ell] = W[g+1.. g+\ell]$ holds by the periodicity.
  This contradicts the uniqueness of $x$.

  We have shown that the maximum number of 
  {SUPSs whose center} 
  is at most $p$, is two.
  Symmetrically, the maximum number of 
  {SUPSs whose center} 
  is at least $p$ is also two.
  Thus, the maximum number of SUPSs for a single query interval is four.

  Finally, we show that the upper bound is tight.
Let us consider the following string $S \in \{\mathtt{a}, \mathtt{b}\}^\ast$ of length $87$:
  \begin{align*}
    S = &\mathtt{aababaaababaaabab{\color{blue}a}aabaaabaaabaaabaaa}&&\backslash\backslash\text{ length 36}\\
        &+ \mathtt{ababaaababaaababa}&&\backslash\backslash\text{ length 17}\\
        &+ \mathtt{baaababaaababaaab}&&\backslash\backslash\text{ length 17}\\
        &+ \mathtt{baaabaaabaaabaaab}&&\backslash\backslash\text{ length 17}.\\
  \end{align*}
  The operator $+$ denotes the concatenation of strings.
For this string and query interval $[18, 18]$ (highlighted in blue in the figure),
  $\SUPS_S([18, 18]) = \{[1, 19], [4, 22], [16, 34], [18, 36]\}$ holds.
  Note that palindromes $S[2.. 18]$, $S[5.. 21]$, and $S[17.. 33]$, which are shorter than $19$ and cover the interval $[18, 18]$, are not unique
  since each of them has another occurrence in the artificial gadgets concatenated by $+$ operators.
  Also, it can be easily checked that all palindromes of length at most $18$ that cover the interval $[18, 18]$ are not unique.
\end{proof}

The above example having four SUPSs is of length $87$, and the length of each SUPS is $19$.
The smallest period of the former two SUPSs is $6$, and that of the latter two SUPSs is $4$.
We do not know if the example is the shortest one.
 \subsection{Sum of Lengths of All MUPSs}\label{subsec:sum}
In this subsection,
we show an upper bound of the total sum of the 
{lengths} 
of MUPSs in a string $T$.
It was shown that
the number of MUPSs stabbed by a single position is $O(\log \abs{T})$ in~\cite{funakoshiMUPSafteredit}.
This immediately implies that
the sum of the lengths of MUPSs of a string $T$ is $O(\abs{T}\log \abs{T})$. 
Here, we improve the upper bound as follows:
\begin{theorem}\label{thm:sum_of_MUPSlen}
  The total sum of the lengths of MUPSs of a string $T$ is $O(\abs{T})$. 
\end{theorem}

In order to show Theorem~\ref{thm:sum_of_MUPSlen},
we first analyze the sum of the lengths of MUPSs covering a single position.
Let $p$ be an arbitrary position in a string $T$.
Also, let $\lMUPS_T(p)$ be the set of MUPSs that cover $p$ and whose centers are at most $p$.
Namely, $\lMUPS_T(p) = \{[s, t]\in \MUPS(T)\mid s \le p \le t\text{ and }(s+t)/2 \le p\}$.
We show the following lemma:
\begin{lemma}\label{lem:sum_of_leftMUPS}
  For any position $p$ in a string $T$,
  $$\sum_{[s,t]\in\lMUPS_T(p)}(t-s+1) \in O(L_p)$$
  holds where $L_p$ is the maximum length of MUPSs covering position $p$.
\end{lemma}
\begin{proof}
  {Each} MUPS in $\lMUPS_T(p)$ is an expansion of some palindromic suffix of $T[1.. p]$.
  For a set $G$ of palindromic suffixes of $T[1.. p]$ and a MUPS $\mu$ of $T$,
  we say that $\mu$ is \emph{obtained from} $G$ if $\mu$ is an expansion of some palindrome in $G$.
{Here we use some notations used in Section~\ref{subsec:period}.}
  {
    Namely, $G_1, G_2, \ldots, G_{\pi}$ are the groups of palindromic suffixes of $T[1..p]$.
    Also, $k$ is the largest index of groups such that $G_k$ is of type 2.}
  Since no MUPS can be obtained from 
  {$\bigcup_{i\in[1, k-1]} G_i \setminus\{u_k, u_kv_ku_k\}$} 
  by Lemma~\ref{lem:type2_expansion}, we only consider the set $\bigcup_{j\in[k+1, \pi]}G_j \cup G_k \cup \{u_k, u_kv_ku_k\}$.
{If there are no MUPSs obtained from $\bigcup_{j\in[k+1, \pi]}G_j$,
  the number of MUPSs obtained from $\bigcup_{j\in[k+1, \pi]}G_j \cup G_k \cup \{u_k, u_kv_ku_k\}$ is $O(1)$ and the lemma holds.
  Thus, in the following, we consider the case that there exists some MUPSs obtained from $\bigcup_{j\in[k+1, \pi]}G_j$.
}
  
  Now, let $\tilde{\mathcal{G}} = \tilde{G}_1, \ldots, \tilde{G}_{\ell}$
  be the subsequence of the sequence $G_{k+1}, \ldots, G_{\pi}$ such that for each $\tilde{G}_i \in \tilde{\mathcal{G}}$, there exists a MUPS obtained from $\tilde{G}_{i}$.
  Further let $\tilde{s}_{i}$ be the shortest palindrome in $\tilde{G}_{i}$, and let $\tilde{d}_i$ be the smallest period for $\tilde{G}_{i}$.
  Then, by Lemma~\ref{lem:type1_expansion}, for every $i$ with $1 \le i \le \ell-2$,
  any MUPS obtained from group $\tilde{G}_i$ is not longer than $\tilde{s}_{i+2}$.
Since there are at most two MUPSs obtained from a single group~\cite{funakoshiMUPSafteredit},
  the sum of lengths of MUPSs obtained from $\tilde{G}_1\cup\ldots\cup\tilde{G}_{\ell-2}$ is bounded by
  $2\tilde{s}_3+2\tilde{s}_4+\cdots + 2\tilde{s}_\ell$. In general, $\tilde{s}_r < 2\tilde{d}_r$ holds from the periodicity.
  Thus, $2\tilde{s}_3+2\tilde{s}_4+\cdots + 2\tilde{s}_\ell < 4\sum_{r=3}^{\ell} \tilde{d}_r \in O(\tilde{d}_\ell)$
  since $\tilde{d}_{r+2} \geq \tilde{d}_{r+1} + \tilde{d}_{r}$ holds for any $3 \leq r \leq \ell-2$ by Lemma~\ref{lem:palindromic_suffixes}.
  Let $L_p'$ be the maximum length of MUPSs obtained from $H = \{u_k, u_kv_ku_k\}\cup G_k\cup\tilde{G}_{\ell-1}\cup\tilde{G}_{\ell}$.
  Then, the sum of lengths of MUPSs obtained from $H$ is $O(L_p')$.
  Therefore, the total sum of the lengths of MUPSs in $\lMUPS_T(p)$ is in $O(\tilde{d}_\ell + L_p')\subseteq O(L_p)$
  since $\tilde{d}_\ell \le \tilde{s}_\ell \le L_p' \le L_p$.
\end{proof}

By symmetry, Lemma~\ref{lem:sum_of_leftMUPS} immediately leads to the next Corollary~\ref{cor:sum_of_stabbed_MUPSlen}.
\begin{corollary}\label{cor:sum_of_stabbed_MUPSlen}
  For any position $p$ in a string $T$,
  the sum of the lengths of MUPSs covering position $p$ is $O(L_p)$
  where $L_p$ is the maximum length of MUPSs covering position $p$.
\end{corollary}

We are now ready to prove Theorem~\ref{thm:sum_of_MUPSlen}.
\begin{proof}[Proof of Theorem~\ref{thm:sum_of_MUPSlen}]
  {We divide the set of MUPSs into two sets.}
  First, let $\mathit{ML}$ and $\mathit{MS}$ be the empty sets initially.
  We then perform the following operations for each MUPS of $T$
  in descending order of their lengths (the order of two elements of the same length is arbitrary):
  If $[b,e] \in \MUPS(T)$ does not share the same position for any $[b',e'] \in \mathit{ML}$, then update $\mathit{ML} = \mathit{ML} \cup \{[b,e]\}$.
  Otherwise, add $[b,e]$ to $\mathit{MS}$. 
  Then, $\sum_{[b,e]\in\mathit{ML}}{(e-b+1)} \leq \abs{T}$ holds since all elements in $\mathit{ML}$ do not overlap each other.
  Also, any MUPS $[b,e] \in \mathit{MS}$ contains some position $p$ such that $b' \leq p \leq e'$ with $[b',e'] \in \mathit{ML}$.
  Moreover, since MUPSs cannot be nested,
  each MUPS in $\mathit{MS}$ contains either or both of
  the ending position of MUPS $[b'_i,e'_i] \in \mathit{ML}$ and
  the beginning position of MUPS $[b'_{i+1},e'_{i+1}] \in\mathit{ML}$
  for some $i$.
By Corollary~\ref{cor:sum_of_stabbed_MUPSlen}, 
  the sum of the length of MUPSs covering position $b'_i$ or $e'_i$ is $O(e'_i-b'_i+1)$.
  Therefore, $\sum_{[b,e]\in\mathit{MS}}{(e-b+1)} \in O(\sum_{[b',e']\in\mathit{ML}}{(e'-b'+1)}) \subseteq O(\abs{T})$ holds (see also Figure~\ref{fig:sum_of_MUPSs}).
  This completes the proof.
\end{proof}
\begin{figure}[t]
  \centerline{
    \includegraphics[width=0.8\linewidth]{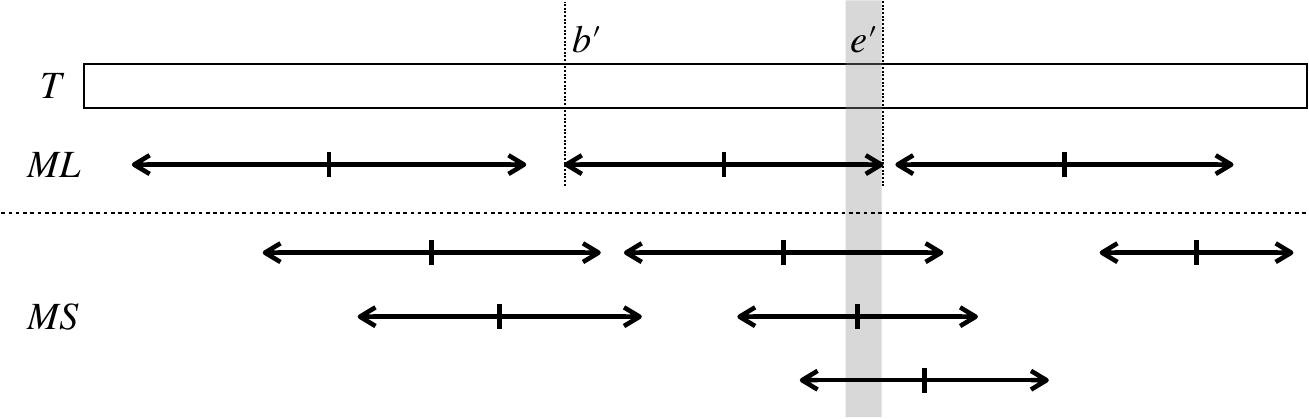}
  }
  \caption{
    Illustration for the proof of Theorem~\ref{thm:sum_of_MUPSlen}.
    There are three MUPSs in $\mathit{ML}$ and six MUPSs in $\mathit{MS}$.
    Also, there are four MUPSs that cover the ending position $e'$ of MUPS $[b', e'] \in \mathit{ML}$.
    Since $[b', e']$ is a longest one among them,
    the sum of their lengths is bounded by $O(e' - b'+1)$ by Corollary~\ref{cor:sum_of_stabbed_MUPSlen}.
  }
  \label{fig:sum_of_MUPSs}
\end{figure}

Using Theorem~\ref{thm:sum_of_MUPSlen}, we design a simple algorithm for computing all point SUPSs.
As in Inoue et al.'s method described in Section~\ref{sec:static_algo},
we use the MUPSs and maximal palindromes to compute point SUPSs.
However, we can avoid using an RmQ data structure.
\begin{proposition}\label{prop:simple}
  Given the set of MUPSs of a string $T$ of length $n$ and the set of maximal 
  {palindromes} 
  of $T$
  one can compute all point SUPSs for all positions in $O(n)$ time without using RmQs.
\end{proposition}
\begin{proof}
First, for each position, we record the shortest MUPS(s) covering the position.
  If there are no such MUPSs for a position, we record $\infty$ for the position.
  This can be done in $O(n)$ time by scanning all MUPSs naively since Theorem~\ref{thm:sum_of_MUPSlen} holds.
Next, for each position $p$, compare the three following values and find the shortest one(s):
  (1) the shortest expansion of the left-neighbor MUPS of $p$,
  (2) the recorded value (i.e., the length of the shortest MUPS covering $p$), and
  (3) the shortest expansion of the right-neighbor MUPS of $p$.
  Note that (1) and (3) may not exist.
  This can be done in a total of 
  {linear} 
  time
  by using an $O(n)$-space and $O(1)$-query time predecessor/successor data structure.
\end{proof}
 \section{Compact SUPS Data Structures} \label{sec:compact}
In this section, we propose space-efficient SUPS data structures.
{To our knowledge, the only SUPS data structure that can be sublinear size,
i.e., $o(n \log n)$ bits, is that of Watanabe et al.~\cite{WatanabeNIBT20}.}
Watanabe et al.~\cite{WatanabeNIBT20}
proposed a SUPS data structure of size {$O(r\log n)$ bits}
where $r \le n$ is the size of the run-length encoded string.
Their data structure will be small when the input string is highly compressible with run-length encoding.
On the other hand, {$O(r\log n)$ bits} can be large as much as {$O(n\log n)$ bits} in the worst case.
In this section, we propose $O(n)$-bits data structures which can answer SUPS queries in optimal time.
Namely, {our data structure is} always space-efficient
regardless of the compression scheme or characteristic structures of the input string.
In the rest of this paper, let $m$ be the number of MUPSs of string $T$.
The sizes of our data structures are $3n+2m+o(n)$ bits for interval SUPS queries, and $3n+m'+o(n)$ bits for point SUPS queries
where $m' \le m$ is the number of \emph{meaningful MUPSs} that we will define later.

\subsection{{Data Structures for} Interval SUPS Queries}
First, we show a compact representation of the data structure of Inoue et al.
Our data structure consists of compact representations of
(1) the set of MUPSs, (2) the set $\mathcal{M}$ of maximal palindromes each of which is an expansion of some MUPS, and
(3) an RmQ data structure over the sequence of the lengths of MUPSs.
\begin{itemize}
  \item[(1)]
    We represent the set of MUPSs as two length-$n$ bit-arrays $B$ and $E$ that indicate the beginning and the ending positions of MUPSs.
    Namely, $B[i] = \mathtt{1}$ iff some MUPS begins at $i$, and  $E[j] = \mathtt{1}$ iff some MUPS ends at $j$ for each $1 \le i, j \le n$.
    Since MUPSs cannot be nested, the number of the set-bits in $B$ is exactly $m$, and in $E$ as well~(see Fig.~\ref{fig:BEL} for examples).
  \item[(2)]
    We represent $\mathcal{M}$ as the length-$n$ bit-array $L$ that indicates the beginning positions of maximal palindromes in $\mathcal{M}$.
    Namely, $L[i] = \mathtt{1}$ iff some palindrome in $\mathcal{M}$ begins at $i$ for each $1 \le i \le n$.
    Since all palindromes in $\mathcal{M}$ are unique by the definition, they cannot be nested,
    and thus, the number of the set-bits in $L$ is exactly $m$.
    Namely, the $i$-th set-bit in $L$ corresponds to the $i$-th MUPS.
    Note that, for any $k$, we can restore the ending position of the $k$-th palindrome in $\mathcal{M}$ from three arrays $B$, $E$, and $L$,
    that is, $b + e - \ell$ where $b$, $e$, and $\ell$ are the positions of the $k$-th set-bits in $B$, $E$, and $L$, respectively.
  \item[(3)]
    We build the succinct RmQ data structure of~\cite{rmqsuccinct} on the sequence of the lengths of MUPSs.
    The size of the data structure is $2m+o(m)$ bits.
\end{itemize}
Also, we enhance three bit-arrays $B$, $E$, and $L$ with rank/select dictionaries.
Then, we can completely simulate the algorithm of Inoue et al., i.e., any SUPS query can be answered in constant time.
This data structure requires $3n + 2m + o(n)$ bits of space.

{The construction time is linear: All the maximal palindromes in $T$ and all the MUPSs of $T$ can be computed in $O(n)$ time~\cite{manacher1975new,inoue2018algorithms}, and hence, three bit-arrays $B$, $E$, and $L$ can be computed in $O(n)$ time. Also, the rank/select dictionaries and the succinct RmQ data structure for a bit-array of length $n$ can be constructed in $O(n)$ time~\cite{Jacobson89,clark1997compact,rmqsuccinct}.}
\begin{figure}[t]
  \centerline{
    \includegraphics[width=0.6\linewidth]{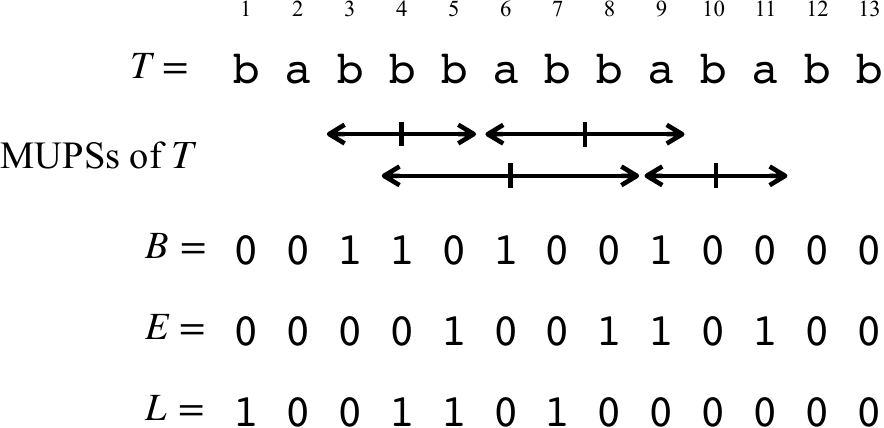}
  }
  \caption{
    Three arrays $B$, $E$, and $L$ for string $T = \mathtt{babbbabbababb}$. 
    The first two arrays $B$ and $E$ indicate the beginning and ending positions of MUPSs of $T$.
    The third array $L$ indicates the beginning positions of the maximal expansions of MUPSs of $T$.
    For instance, $L[1] = \mathtt{1}$ holds
    since $T[1.. 7] = \mathtt{babbbab}$ is the maximal palindrome centered at position $4$, which is the center of MUPS $T[3.. 5]$.
  }
\label{fig:BEL}
\end{figure}

To summarize, we obtain the next theorem:
\begin{theorem}
  There is a data structure of size $3n + 2m + o(n)$
  that can answer any interval SUPS query in $O(1)$ time
  where $m$ is the number of MUPSs of $T$.
  Also, given $T$, we can construct the data structure in $O(n)$ time.
\end{theorem}

\subsection{{Data Structures for} Point SUPS Queries}
As for the point SUPS queries, we can further reduce the space usage with a new algorithm specialized to point queries.
Firstly, we define a new\footnote{This is inspired by a similar notion defined for minimal unique substrings in \cite{Tsuruta2014}.} notion for MUPSs.
\begin{definition}
  A MUPS is said to be meaningful if there is a SUPS which is an expansion 
  {of the respective MUPS} 
  for some position.
  {MUPSs which are not meaningful are said to be meaningless.}
\end{definition}
For example, in Fig.~\ref{fig:BEL},
MUPS $T[4.. 8] = \mathtt{bbabb}$ of length $5$ is meaningless
since the lengths of SUPSs for positions $4$ and $5$ are $3 < 5$ and 
the lengths of SUPSs for positions $6$, $7$, and $8$ are $4 < 5$.
Given the set of MUPSs, we can compute the set of meaningful MUPSs in $O(n)$ time
by computing all SUPSs for all positions (e.g.,~Proposition~\ref{prop:simple}) and removing MUPSs unused.

Let $\MLen = (x_1, \ldots, x_{m'})$ be the sequence of the lengths of meaningful MUPSs sorted in increasing order on their starting positions.
Also, let $\MLen_p \subseteq \MLen$ be the sequence of the lengths of meaningful MUPSs stabbed by a position $p$.
The next lemma states that $\MLen_p$ has a sort of monotonicity.
\begin{lemma}\label{lem:vshape}
  There are no three elements $x_i$, $x_j$, and $x_k$ in $\MLen_p$
  such that $i < j < k$ and $x_i < x_j > x_k$.
\end{lemma}
\begin{proof}
  Assume on the contrary that there exist $x_i$, $x_j$, and $x_k$ satisfying the conditions above.
  Let $s_i$, $s_j$, and $s_k$ be the starting positions of the MUPSs, respectively.
{Since MUPSs cannot be nested and three the MUPSs cover the same position $p$,}  every position inside the second MUPS is covered by the first MUPS or the third MUPS.
  Namely, $s_i < s_j$, $s_j+x_j-1 < s_k+x_k-1$, and $s_k \le p \le s_i+x_i-1$ hold~(see Fig.~\ref{fig:meaninglessMUPS}).
  For each position $q \in [1, s_i]$, there is no expansion of the second MUPS $T[s_j.. s_j+x_j-1]$ starting at $q$
  because if such a palindrome exists, it contradicts the uniqueness of the first MUPS $T[s_i.. s_i+x_i-1]$.
  Symmetrically, for each position $q' \in [s_k+x_k-1, n]$, there is no expansion of the second MUPS $T[s_j.. s_j+x_j-1]$ ending at $q'$.
  Finally, for each 
  {position} 
  $q'' \in [s_i+1, s_k+x_k-2]$,
  any palindrome covering both $q''$ and $[s_j, s_j+x_j-1]$ cannot be a SUPS for $q''$
  since the second MUPS is longer than another MUPS covering $q''$.
  Thus, the second MUPS is meaningless, a contradiction. 
\end{proof}
\begin{figure}[t]
  \centerline{
    \includegraphics[width=0.8\linewidth]{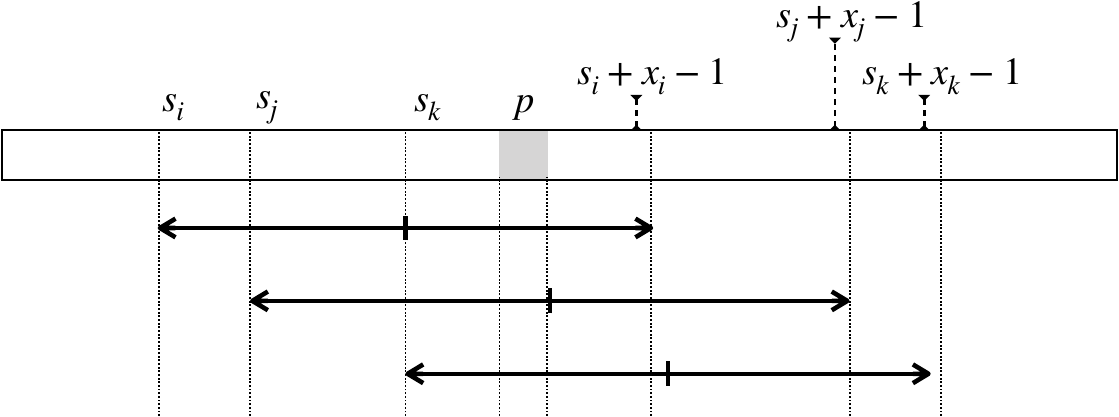}
  }
  \caption{
    Illustration for a contradiction in the proof of Lemma~\ref{lem:vshape}.
    The second MUPS which is the longest among three MUPSs cannot be a meaningful MUPS.
  }
\label{fig:meaninglessMUPS}
\end{figure}

From this lemma, $\MLen_p$ can be regarded as a concatenation of
a (possibly empty) non-increasing sequence and a non-decreasing sequence.
Thus, if we find the leftmost value $x_t$ such that $x_t > x_{t-1}$,
then $x_{t-1}$ is the smallest in $\MLen_p$ and $x_{t'}$ is not for each $t' \ge t$.
Also, from the monotonicity of $\MLen_p$ before $x_t$,
we can find MUPSs of length {equal to} $x_{t-1}$
by (backward) linear search starting at $x_t$ on $\MLen_p$.
In order to find such $x_t$, we precompute the bit-array $\inc$ of length $\abs{\MLen}$
such that $\inc[0] = \mathtt{0}$, and $\inc[i]=\mathtt{1}$ iff $x_i>x_{i-1}$.
Also, we enhance the array $\inc$ with rank/select dictionaries so that
we can answer any successor query on $\inc$ in constant time.

Our data structure includes the bit-arrays $B$, $E$, and $L$ as in the previous subsection.
On the other hand, instead of an RmQ data structure of size $2m+o(m)$ bits,
we use bit-array $\inc$ of length $m'$ and rank/select dictionaries of size $o(m')$
where $m' = \abs{\MLen} \le m$.

Our algorithm is almost the same as Inoue et al.'s one
except that 
{we use the bit-array} 
$\inc$ instead of an RmQ data structure.
Given a query position $p$,
we first compute the left-neighbor and right-neighbor MUPSs
and check whether their expansion can cover $p$ or not.
Simultaneously, we obtain the range of meaningful MUPSs stabbed by $p$.
Let $i$ (resp.,~$j$) be the index of the leftmost (resp.,~rightmost) meaningful MUPS stabbed by $p$.
We then find the position $k$ of the leftmost set-bit in $\inc[i+1.. j]$.
{If such a set-bit does not exist (i.e., $\inc[i+1.. j] = \boldsymbol{0}$), let $k = j+1$ for convenience.}
{By the definition of $\inc$ and $k$,
the prefix $(x_{i}, \ldots, x_{k-1})$ of $\MLen_p$ is non-increasing.
Especially, if $k = j+1$, then $\MLen_p = (x_{i}, \ldots, x_{j})$ is non-increasing, and hence, $x_{j}$ is the smallest within $\MLen_p$.
Since there are at most four SUPSs (Theorem~\ref{thm:num_of_SUPS}),
it suffices to compare $x_{j+1-t}$ for $t = 1, 2, 3, 4$.
If $k \le j$, then the suffix $(x_{k}, \ldots, x_j)$ of $\MLen_p$ is non-decreasing by Lemma~\ref{lem:vshape}.
In this case, $x_{k-1}$ is the smallest since $x_{k-1} < x_k$.
Again,} since there are at most four SUPSs (Theorem~\ref{thm:num_of_SUPS}),
it suffices to compare $x_{k-t}$ for $t = 1, 2, 3, 4$.
Therefore, we can find all MUPSs that are candidates for SUPSs for $p$ without using any RmQ data structure.

All the above operations can be performed in constant time
by using rank/select dictionary on $\inc$, $B$, $E$, and $L$.
We obtain the next theorem:
\begin{theorem}
  There is a data structure of size $3n + m' + o(n)$
  that can answer any point SUPS query in $O(1)$ time
  where $m'$ is the number of meaningful MUPSs of $T$.
  Also, given $T$, we can construct the data structure in $O(n)$ time.
\end{theorem}

 \section{Semi-dynamic SUPS Data Structures}\label{sec:dynamic_algo}
In this section, we introduce SUPS data structures
under two semi-dynamic models:
the sliding-window model and the after-edit model.
Our results are based on the static method proposed by Inoue et al.~\cite{inoue2018algorithms},
{which we reviewed in Section~\ref{sec:static_algo}.}

\subsection{Sliding-window Data Structures}

We 
{make} 
some modifications to 
{the static data structures from Section~\ref{sec:static_algo}}
to answer any SUPS queries for a sliding window.

It is shown in \cite{mieno2021eertree} that
the number of changes of MUPSs is constant when we append a character or delete the first character,
and we can detect the changes in amortized $O(\log\sigma)$ time.
Further, predecessor and successor data structures on the MUPSs can be updated dynamically
in $O(\log\log n)$ time using van Emde Boas trees~\cite{vEBtree1977}.

For a dynamic RmQ data structure, we can use the one proposed by Brodal et al.~\cite{brodal2011dynamic}.
However, if we directly apply their data structure to our problem,
the updating time is in $\Omega(\log n / \log\log n)$, and it becomes a bottleneck.
In order to avoid such a situation, we use another dynamic data structure with some constraints which 
{suffices} 
for our problem.

As in the algorithm described {in Section~\ref{sec:static_algo}},
we will use RmQ on the sequence of the lengths of MUPSs.
The width of a query range of RmQ is bounded by the number of MUPSs covering query interval $[p, q]$~(see also Fig.~\ref{fig:sups_cand}).
It is known that the number of MUPSs covering any interval is $O(\log n)$~\cite{funakoshiMUPSafteredit},
hence the width of a query range of RmQ is also $O(\log n)$.
We call the range minimum query such that the width of any query is constrained in $O(\polylog(n))$ {\LogRmQ}.
Later, we show the following lemma:

\begin{lemma}\label{lem:dynamic_logrmq}
  There exists a linear size data structure for a dynamic array $A$
  that supports any {\LogRmQ} on $A$ in constant time.
  We can maintain the data structure in constant time
  when an element of $A$ is substituted by another value.
  Also, we can maintain the data structure in amortized constant time
  when some element is inserted to (or deleted from) $A$.
\end{lemma}

Finally, we show that the set of maximal palindromes for a sliding window can be maintained efficiently.
We generalize Manacher's algorithm \cite{manacher1975new} to the sliding-window model.

\subsubsection{Manacher's Algorithm for Sliding Window.}
Manacher's algorithm is an online algorithm that computes the set of maximal palindromes in a string.
In this subsection, we apply Manacher's algorithm to the sliding-window model.
The problem 
{was} 
solved in \cite{Gawrychowski2019}, however, we will describe a sliding-window algorithm for completeness.

Important invariants of Manacher's algorithm before reading the $i$-th character are
(1) we know the center position $c$ of the longest palindromic suffix of $T[1..i-1]$, and
(2) we know all the maximal palindromes, each of whose center is at most $c$.
Note that for the SUPS query,
we are interested in maximal palindromes, which are \emph{unique} in the string.
Since any palindrome whose center is greater than $c$ is not unique,
the second invariant is sufficient for our purpose.

When a character is appended to the current window,
we update the set of maximal palindromes in the online manner of the original Manacher's algorithm.
When the first character of the window is deleted,
we do not need to do anything if the window $T[b.. e]$ itself is not a palindrome.
Instead, when we refer to the arm-length of the maximal palindrome centered at a specified position,
we need to consider that the left-end of the palindrome may exceed the left-end of the window.
Namely, if the stored arm-length for center $x$ is $\ell_x$, the actual arm-length is $\min\{\ell_x, \ceil{x-b}\}$.

If the window $T[b.. e]$ itself is a palindrome,
we need to update the longest palindromic suffix to keep the first invariant.
This can be done in amortized $O(1)$ time as in Manacher's algorithm.
More precisely, for every (half) integers $j = 0.5, 1, 1.5 \ldots$,
the arm-length of the maximal palindrome of center $\frac{b+e}{2} + j$ is equal to that of center $\frac{b+e}{2} - j$.
Thus, we copy them for incremental $j$'s until we find a palindromic suffix of $T[b+1.. e]$.
Then, we set $c$ to the center position of the suffix palindrome we found.
Since the sequence of center positions of the longest palindromic suffixes of the windows is non-decreasing
while running the algorithm, the total processing time is $O(n)$.

Therefore, we obtain the following:
\begin{theorem}
  There exists a data structure of size $O(W)$ for the sliding-window SUPS problem
  that supports 
    $\sups([p, q])$ in $O(\log\log W)$ time and
    $\swpushback(c)$ and $\swpop()$ in amortized $O(\log\sigma + \log \log W)$ time,
  where $W$ is the size of the window.
\end{theorem}

\subsection{After-edit Data Structure}

In this subsection, 
we design a SUPS data structure for the after-edit model.
Basically, the idea is the same as the previous one.
The only difference is that we do not maintain maximal palindromes in the {after-substitution SUPS} problem.
Instead, we use a bidirectional LCE on the original string $T$.

\begin{theorem}
  There exists a data structure of size $O(n)$ for the {after-substitution SUPS} problem
  that can be updated in amortized $O(\log\sigma + (\log\log n)^2 + d\log\log n)$ time for a single substitution and
  can answer any subsequent SUPS queries in $O(k\log\log n)$ time,
  where
    $d$ is the number of changes of MUPSs when the substitution is applied to $T$, and
    $k$ is the number of the SUPS queries after the substitution.
  Also, given a string $T$, the data structure can be constructed in $O(n)$ time.
\end{theorem}
\begin{proof}
  Given a substitution operation, we can detect all the changes of MUPSs
  in $O(\log\sigma + (\log\log n)^2 + d)$ time~\cite{funakoshiMUPSafteredit}.
  Then,
  the set of MUPSs can be updated in $O(d)$ time,
  the predecessor/successor data structures can be updated in $O(d\log\log n)$ time, and
  the {\LogRmQ} data structure can be updated in amortized $O(d)$ time by Lemma~\ref{lem:dynamic_logrmq}.
Finally, we can compute the maximal palindromes in $T'$ that are expansions of the left-neighbor and the right-neighbor MUPSs
  by answering a constant number of bidirectional LCE queries on $T$ while skipping the edited position (so-called kangaroo jumps).
Also, it is known that the set of MUPSs of $T$, the predecessor/successor data structures, and the LCE data structure can be computed in $O(n)$ time.
  Further, the {\LogRmQ} data structure can be computed in $O(n)$ time by Lemma~\ref{lem:dynamic_logrmq}.
\end{proof}

 \subsection{Dynamic LogRmQ} \label{subsec:dynamic_rmq}

In this subsection, we give a proof of Lemma~\ref{lem:dynamic_logrmq}.
We assume that the width of the query range is constrained in $O(\log^c n)$ for a fixed constant $c$.
  We first consider dividing the input array $A$ into blocks of size $\log^cn$.
We call each of the blocks large block.
Then, we build a linear size dynamic RmQ data structure on each large block.
Given a query range of width
  $O(\log^c n)$,
we get range minima from a constant number of large blocks and then naively compare them.

We update the RmQ data structure on the large block containing the edited position
when the input array $A$ is edited.
If the size of a large block becomes far from
  $\log^c n$
by insertions or deletions,
then we split a block or merge continuous blocks to keep the size in
  $\Theta(\log^c n)$.
For example, we split a block into two blocks when the block size exceeds
  $2\log^c n$
and
merge two adjacent blocks when the block size falls below
  $\frac{1}{2}\log^c n$.
If each large block can be updated in amortized constant time,
the whole data structure can also be updated in amortized constant time.
In the next subsection, we consider how to treat a large block.

\paragraph{{\bf Recursive Structure of Large Block.}}

In order to update large blocks efficiently, we apply the path minima data structure proposed by Brodal et al.~\cite{brodal2011dynamic}.
They treated the problem of path minima queries on a tree, a generalization of range minimum queries on an array.

First, we divide a large block $B$ of length $\Theta(\log^c n)$ into small blocks each of length $L = \Theta(\log^\varepsilon n)$
where $\varepsilon < 1$ is an arbitrary small constant.

Let $B_1$ be the array of length
  $\Theta(\log^{c-\varepsilon} n)$
that stores the minima of small blocks on $B$.
A query on large block $B$ can be reduced to at most two queries on small blocks and at most one query on $B_1$ (see Fig.~\ref{fig:rmq}).
\begin{figure}[t]
  \centerline{
    \includegraphics[width=0.7\linewidth]{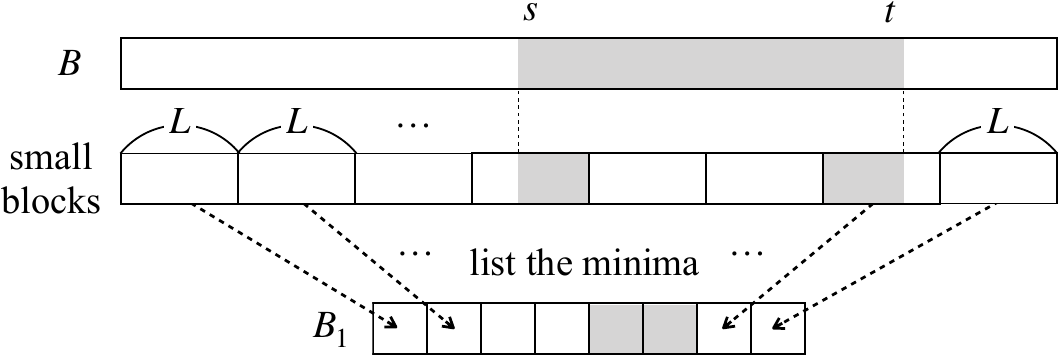}
  }
  \caption{
    Illustration for dividing a large block $B$ into small blocks.
    Query $[s, t]$ on $B$ can be reduced to 
    queries inside the fourth and the seventh small blocks,
    and query $[5, 6]$ on $B_1$.
  }
  \label{fig:rmq}
\end{figure}
Similarly, for every $i \ge 2$, we divide $B_{i-1}$ into small blocks of the fixed-length $L$
and let $B_i$ be the array of size
  $\Theta(\log^{c-i\varepsilon} n)$
that stores the minima of small blocks on $B_{i-1}$.
A query on $B_i$ can be reduced to at most two queries on small blocks and at most one query on $B_{i+1}$ at the next level.
We recursively apply such division until the size of $B_i$ becomes a constant.
The recursion depth is
  $O(c/\varepsilon)$,
i.e., a constant.
Notice that recursion occurs at most once at each level.
Thus, if we answer RmQ inside a small block in constant time, then the total query time is also a constant.

We answer a query on each small block by using a lookup-table,
where the index is a pair of a small block and a query range, and the value is the answer (the position of a minimum).
Since RmQ returns the \emph{position} corresponding to a range minimum,
we can convert each small block 
{to a sequence of its} 
\emph{local ranks}.
Then, the number of possible variants of such small blocks is at most 
$O((\log^\varepsilon n)^{\log^\varepsilon n}) \subset o(n)$.
Also, the total variations with all possible query intervals are still $O((\log^\varepsilon n)^2)$ $\subset o(n)$,
i.e., the number of elements in the lookup-table is $O((\log^\varepsilon n)^{\log^\varepsilon n+2}) \subset o(n)$.
Furthermore, a small block (i.e., an element in the lookup-table) can be represented in $o(\log n)$ bits:
the length, the pointers to each element, and the local ranks.
Thus, table lookup can be done in constant time.
Namely, the time complexity of an RmQ on a small block is constant.
For substitutions (resp., insertions and deletions), updating small blocks can be done in worst-case (resp., amortized) constant time
by combining another lookup-table and Q-heap (cf.~\cite{brodal2011dynamic}).
Therefore, we have proven Lemma~\ref{lem:dynamic_logrmq}.
 \section{Conclusions and Discussions}\label{sec:conclusions}

In this paper, we studied SUPS problems 
{of} 
static and non-static strings.
Firstly, we showed combinatorial properties on the problems;
the tight upper bound on the maximum number of SUPSs for a single interval, and
the sum of lengths of MUPSs of a string is linear to the length of the string.
Secondly, we improved Inoue et al.'s time-optimal algorithm on space usage,
i.e., we designed a compact data structure of size $3n+2m+o(n)$ bits that can answer any interval SUPS query in constant time
where $n$ is the length of the input string and $m$ is the number of MUPSs of the input string.
Also, we proposed a new method specialized for the point SUPS problem,
and based on the method,
we designed a more space-efficient compact data structure of size $3n+m'+o(n)$ bits that can answer any point SUPS query in constant time
where $m'$ is the number of meaningful MUPSs of the input string.
Finally, we considered SUPS problems in two semi-dynamic models.
{We} 
proposed a data structure of size $O(W)$ for the sliding-window SUPS problem
that supports any SUPS query and window-shift operation in $\tilde{O}(1)$ time
where $W$ is the size of the window.
Further, 
we propose a data structure of size $O(n)$ for the {after-substitution SUPS} problem
that can answer any SUPS query after a single character substitution
in amortized $\tilde{O}(1)$ time.
As a by-product, we proposed a fully-dynamic data structure for the range minimum queries
in which the width of each query range is in $\tilde{O}(1)$.

Designing an efficient SUPS algorithm for a fully dynamic setting is future work.
All the known algorithms for SUPS queries for static/non-static strings 
basically precompute the set of MUPSs of the input string.
If we try to extend such algorithms to a fully-dynamic one,
maintaining the set of MUPSs may be a 
{bottleneck.}
To the best of our knowledge,
there is no study that deals with unique substrings in a dynamic string
while maintaining palindromic structures in a fully dynamic string has been studied in some literature~\cite{Amir_Boneh_19,Amiretal20}.
In general, occurrences of substrings 
{can change dramatically} 
when a string is edited.
The algorithm of~\cite{funakoshiMUPSafteredit} can capture the changes of MUPSs for a single edit
after linear-time preprocessing, however, it cannot be applied directly to multiple edits.
Another possible way is to compute SUPSs for a query interval without using MUPSs,
namely, to determine the uniqueness of palindromes covering the query interval in a dynamic string.
For an implementation of this idea, a dynamic suffix array proposed by Kempa and Kociumaka~\cite{dynamicSA} might be useful.

\section*{Acknowledgements}
We would like to thank Professor Jeffrey Shallit (University of Waterloo)
for his interest in our paper and his advice to simplify our proofs.
We would also like to thank the anonymous referees for their helpful comments on the manuscript.
This work was partially supported by the JSPS KAKENHI Grant Numbers JP20J11983, JP22K21273 (TM), and JP20J21147 (MF).
This preprint has not undergone peer review or any post-submission improvements or corrections.
The Version of Record of this article is published in Algorithmica, and is available online at \url{https://doi.org/10.1007/s00453-023-01170-8}.

\end{document}